\documentclass[a4paper,11pt]{article}
\usepackage[latin1]{inputenc}

\usepackage{fourier}
\usepackage{graphicx}
\usepackage{amsmath,amssymb}
\usepackage{amsthm}
\usepackage{tikz}

\usepackage{anysize}
\usepackage{enumerate}
\usepackage{mdwlist}

\usepackage{todo}

\providecommand{\setZ}{\mathbb{Z}}
\providecommand{\setQ}{\mathbb{Q}}
\providecommand{\setR}{\mathbb{R}}

\usepackage{pgfmath}
\usetikzlibrary{patterns}

\newcommand {\R}	  {\mathbb{R}}
\newcommand {\N}	  {\mathbb{N}}

\newcommand {\op}[1]	  {\operatorname{#1}}

\newcommand {\shortspace} {\vspace{4mm}}

\newcommand {\norm}[1]	  {\left\|#1\right\|}
\newcommand {\infnorm}[1] {\left\|#1\right\|_\infty}
\newcommand {\CVP} {\ensuremath{\op{CVP}}}

\newtheorem {theorem}	  {Theorem} [section]
\newtheorem {lemma}	  [theorem] {Lemma}
\newtheorem {corollary}   [theorem] {Corollary}

\begin{document}

\title{Covering Cubes and the Closest Vector Problem} 
\author{Friedrich Eisenbrand\thanks{E-mail: \texttt{friedrich.eisenbrand@epfl.ch}},
Nicolai Hähnle\thanks{E-mail: \texttt{nicolai.haehnle@epfl.ch}},
Martin Niemeier\thanks{E-mail: \texttt{martin.niemeier@epfl.ch}}\\
EPFL, Switzerland}
\date{\today}
\maketitle

\begin{abstract}

\noindent 
We provide the currently fastest randomized $(1+\varepsilon)$-approximation algorithm for  the
\emph{closest vector problem} in the $\ell_\infty$-norm. The running time of
our method depends on the dimension $n$ and the approximation guarantee
$\varepsilon$ by  $2^{O(n)}(\log1/\varepsilon)^{O(n)}$ which improves upon the
$(2+1/\varepsilon)^{O(n)}$ running time of the previously best algorithm by 
Blömer and Naewe.


Our algorithm is based on a solution of the following geometric
covering problem that is of interest of its own: Given $\varepsilon\in(0,1)$,
how many ellipsoids are necessary to cover the cube
$[-1+\varepsilon, 1-\varepsilon]^n$ such that all ellipsoids are contained in the
standard unit cube $[-1,1]^n$? We provide an almost optimal bound for
the case where the ellipsoids are restricted to be axis-parallel.

We then apply our covering scheme to a variation of this covering
problem where one wants to cover $[-1+\varepsilon,1-\varepsilon]^n$  with parallelepipeds that, if
scaled by two, are still contained in the unit cube.  Thereby, we
obtain a method to boost any $2$-approximation algorithm for
closest-vector in the $\ell_\infty$-norm to a $(1+\varepsilon)$-approximation
algorithm that has the desired running time.



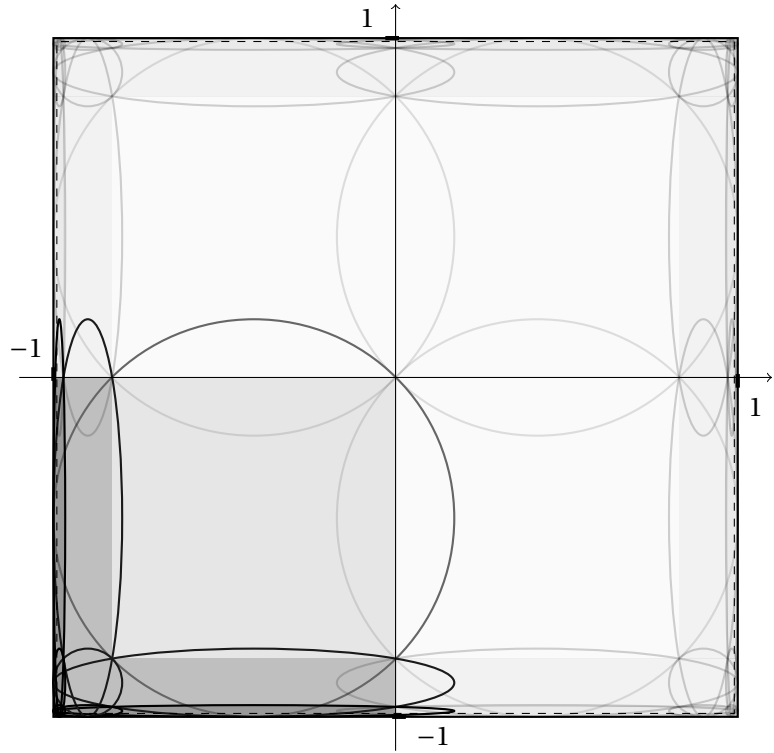
\addtocounter{page}{-1}
\thispagestyle{empty}
\end{abstract}

\shortspace
\begin{center}
  \begin{tikzpicture}[scale=0.45]
    \def\eps{0.1}
    \def\width{10}
    \def\reducedopacity{0.2}

    \pgfmathparse{(1+2/(1.41421356-1))}\pgfmathresult \let\r\pgfmathresult

      \draw[very thick,color=white,fill=white] (-\width,-\width) rectangle  (\width,\width);

     \foreach \x in {0,1,2} {
        \foreach \y in {0,1,2} {           
	   \pgfmathparse{max(\x, \y)*30+20}\pgfmathresult \let\hue\pgfmathresult
           
	   \pgfmathparse{\width*(\r^(-\x-1)) - \width}\pgfmathresult \let\lowerX\pgfmathresult
           \pgfmathparse{\width*(\r^(-\y-1)) - \width}\pgfmathresult \let\lowerY\pgfmathresult
           \pgfmathparse{\width*(\r^(-\x)) - \width}\pgfmathresult \let\upperX\pgfmathresult
           \pgfmathparse{\width*(\r^(-\y)) - \width}\pgfmathresult \let\upperY\pgfmathresult

           \draw[thin,draw=none,fill=gray!\hue] (\lowerX, \lowerY) rectangle (\upperX, \upperY);
	   \draw[thin,draw=none,fill=gray!\hue,fill opacity=\reducedopacity] (-\lowerX, -\lowerY) rectangle (-\upperX, -\upperY);
	   \draw[thin,draw=none,fill=gray!\hue,fill opacity=\reducedopacity] (-\lowerX, \lowerY) rectangle (-\upperX, \upperY);
	   \draw[thin,draw=none,fill=gray!\hue,fill opacity=\reducedopacity] (\lowerX, -\lowerY) rectangle (\upperX, -\upperY);
       }
     }

  \foreach \x in {0,1,2} {
        \foreach \y in {0,1,2} {
           \pgfmathparse{max(\x, \y)*30+60}\pgfmathresult \let\hue\pgfmathresult

           \pgfmathparse{\width*(\r^(-\x-1)) - \width}\pgfmathresult \let\lowerX\pgfmathresult
           \pgfmathparse{\width*(\r^(-\y-1)) - \width}\pgfmathresult \let\lowerY\pgfmathresult
           \pgfmathparse{\width*(\r^(-\x)) - \width}\pgfmathresult \let\upperX\pgfmathresult
           \pgfmathparse{\width*(\r^(-\y)) - \width}\pgfmathresult \let\upperY\pgfmathresult

           \pgfmathparse{(\lowerX+\upperX)/2}\pgfmathresult \let\centerX\pgfmathresult
	   \pgfmathparse{(\lowerY+\upperY)/2}\pgfmathresult \let\centerY\pgfmathresult

	   \pgfmathparse{\centerX+\width}\pgfmathresult \let\radX\pgfmathresult
           \pgfmathparse{\centerY+\width}\pgfmathresult \let\radY\pgfmathresult

           \draw[thick,color=black!\hue] (\centerX, \centerY) ellipse ({\radX}cm and {\radY}cm);
	   \draw[thick,color=black!\hue,draw opacity=\reducedopacity] (-\centerX, -\centerY) ellipse ({\radX}cm and {\radY}cm);
	   \draw[thick,color=black!\hue,draw opacity=\reducedopacity] (-\centerX, \centerY) ellipse ({\radX}cm and {\radY}cm);
	   \draw[thick,color=black!\hue,draw opacity=\reducedopacity] (\centerX, -\centerY) ellipse ({\radX}cm and {\radY}cm);
       }
     }

    \draw[thick] (-\width,-\width) rectangle  (\width,\width);
    \draw[dashed] (-\width+\eps,-\width+\eps) rectangle  (\width-\eps,\width-\eps);

   \draw[->] (-\width-1,0) -- (1+\width,0);
    \draw[->] (0,-\width-1) -- (0,1+\width);
    \draw[ultra thick] (-0.1,-\width) -- (0.3,-\width) node[below right] {$-1$} (-\width,-0.1) -- (-\width,0.3) node[above left] {$-1$};
    \draw[ultra thick] (0.1,\width) -- (-0.3,\width) node[above left] {$1$} (\width,0.1) -- (\width,-0.3) node[below right] {$1$};

  \end{tikzpicture}
\end{center}
\newpage

\section{Introduction}

The \emph{closest lattice vector problem ($\op{CVP}$) } is one of the
central computational problems in the \emph{geometry of numbers}.  Here, one
is given a rational \emph{lattice} $\Lambda(A) = \{ Ax \colon x \in \setZ^n\}$,
$A \in \setQ^{n× n}$ and a \emph{target vector} $t \in \setQ^n$. The task is
to compute a lattice-point in $\Lambda(A)$ that is closest to $t$ w.r.t. a
given norm.
In this paper, we focus on the closest vector problem in the
$\ell_\infty$-norm. CVP in the $\ell_\infty$-norm  is an \emph{integer programming
  problem}: Given a rational polytope $P \subseteq\setR^n$, compute
an integer point inside $P$ or assert that $P$ does not contain any
integer points. 
On the other hand, any integer programming problem as above can be
directly reduced to $\op{CVP}_\infty$ in a lattice in $m$-dimensional space,
where $m$ is the number of inequalities describing the
polytope.~\footnote{To decide if a polytope $P=\{x\in \R^n~:~Ax\leq u\}$, contains an integer point, compute a vector $l<u$ such that $P=\{x\in \R^n~:~l\leq Ax\leq u\}$. By rescaling each row we can wlog assume that $u-l=\mathbf1$.
Now define $t:=\frac{l+u}{2}$ and observe that $P$ contains an integer point iff there is a $v\in\Lambda(A)$ with $\infnorm{v-t}\leq \frac{1}{2}$. This lattice is not necessarily of full rank, but the techniques of this paper -- whose running time, like those of previous algorithms, depends on the ambient dimension -- can be applied.}
Integer programming is one of the most versatile modeling paradigms
with a wide range of applications. Thus  the closest vector
problem in the $\ell_\infty$-norm variant is particularly important.

The development of methods to solve closest-vector and integer
programming problems resulted in many deep discoveries in
geometry and algorithms. Lenstra~\cite{Lenstra83} showed that integer
programming and thus $\op{CVP}_\infty$ can be solved in polynomial time if
the dimension is fixed. His algorithm lay the first planks between the
geometry of numbers and optimization.  For varying $n$, the running
time of his method is $2^{O(n^3)}$ times a polynomial in the binary
encoding length of the input. Kannan~\cite{kannan87} presented
algorithms for these problems whose running-time dependence on $n$ is
bounded by $2^{O(n \log n)}$.  An important step forward in the quest
for a singly-exponential time algorithm was provided by Ajtai et
al.~\cite{AKS01}. They presented a $2^{O(n)}$ randomized algorithm for
the \emph{shortest vector problem} in the $\ell_2$-norm: Given a
lattice, find the shortest nonzero lattice vector. These results have
been generalized for any $\ell_p$-norm by Blömer and
Naewe~\cite{bn2009}.  Micciancio and Voulgaris~\cite{mv10} provided a
deterministic singly-exponential time algorithm both for the shortest
vector problem as well as for the closest vector problem in the
$\ell_2$-norm. Recently Dadush et al.~\cite{DPV10} have shown that the
shortest vector problem w.r.t. any norm  can be solved with a deterministic
singly-exponential time algorithm.

\subsection*{Approximation algorithms} 
\label{sec:appr-algor}

A $(1+\varepsilon)$-\emph{approximation algorithm} for the closest vector
problem computes a lattice vector whose distance to the target vector
$t$ is at most $(1+\varepsilon)$ times the minimum distance $\min\{ \|v - t \|
\colon v \in \Lambda(A)\}$.  The closest vector problem is NP-hard for any
$\ell_p$ norm~\cite{vEB81} and NP-hard to approximate within constant
factors~\cite{aroraPhD} and even almost polynomial
factors~\cite{dkrs03}.  So clearly one cannot expect to have a
\emph{polynomial-time} approximation scheme (PTAS) for closest
vector. An 
interesting problem is however to design exponential-time
approximation algorithms whose running-time dependence on the
approximation guarantee is not too large. 
Ajtai et al.~\cite{aks02} provided a $(1+\varepsilon)$-approximation algorithm
for $\op{CVP}_2$ with a running time of $2^{O(1+{1}/{\varepsilon})n}$.  Blömer
and Naewe~\cite{bn2009} could improve on this and provide  a
randomized $(1+\varepsilon)$-approximation algorithm for the closest vector
problem w.r.t. \emph{any} $\ell_p$ norm that has  a running time
of $(2+{1}/{\varepsilon})^{O(n)}$.

\medskip 
\noindent 
Our \emph{main result} is a randomized $(1+\varepsilon)$-approximation
algorithm for $\op{CVP}_\infty$ whose running time depends on $n$ and $\varepsilon$
by $2^{O(n)} (\log{1/\varepsilon})^{O(n)}$. In fact, we show that any
singly-exponential time constant factor approximation algorithm can be
strengthened  to a $(1+\varepsilon)$-approximation algorithm that, in the end, has
this running time. Using the randomized algorithm of Blömer and
Naewe~\cite{bn2009} to obtain $2$-approximate solutions, we obtain the
desired running time.

\subsection*{The covering technique}
\label{sec:covering-technique}

We now  explain how coverings of the cube by convex bodies
come into play to obtain the complexity result.  Suppose that we
have an algorithm for closest vector in the $\ell_2$-norm and we want to
apply this to (approximately) decide whether the translated
$\ell_\infty$-unit ball $B = \{ x \in \setR^n \colon \|x - t\|_\infty\leq1 \}$ contains a
lattice point in $\Lambda(A)$. More precisely, given an $\varepsilon>0$, we either
want
\begin{enumerate}[i)]
\item to find a lattice point in $B$, \label{item:1}
\item or to assert that the scaled unit ball $B' = \{ x \in \setR^n \colon \|x
  - t\|_\infty\leq1-\varepsilon \}$ does not contain a lattice point.  \label{item:2}
\end{enumerate}
One obvious idea is to determine a set of balls of radius $\varepsilon$ whose
centers lie in $B'$ and whose union covers $B'$. If we then use the
closest-vector algorithm for the $\ell_2$-norm and target-vectors being
the centers of the balls, we can solve the above problem. If one of
the calls to a  closest vector oracle returns a lattice point of
distance at most $\varepsilon$, then we are in case~\ref{item:1}). Otherwise we
are in case~\ref{item:2}). 

This relates to a classical covering problem. 
Erd\H os and Rogers~\cite{er61} (see also~\cite{MR2433777}) showed that
the space $\setR^n$ can be covered by translates of unit spheres in such
a way that no point is covered by more than $O(n \log n)$ spheres.  We
can use this to cover $[-1+\varepsilon, 1 -\varepsilon]^n$ with spheres of radius $\varepsilon$
that then will be contained in $[-1,1]^n$. The Erd\H os and Rogers
technique would yield an upper bound of $O(n \log n)\frac{(2 - 2\varepsilon)^n}{
(\varepsilon/2)^n V_n}$ where $V_n$ is the volume of the $\ell_2$-unit ball. This
yields the bound $(n/\varepsilon)^{O(n)}$ for the number of queries to the
$\op{CVP}_2$-oracle. Certainly, since the ratio of the volume of
the unit cube $[-1,1]$ to the volume of the $\ell_2$-unit ball $\{x \in \setR^n \colon
\norm{x}_2\leq1\}$ is $2^{\Theta(n \log n)}$, we cannot hope
to improve the dependency on the dimension.  
But can we improve the
dependence on $\varepsilon$?

Since an \emph{ellipsoid} is the image of the $\ell_2$-unit-ball $\{x \in
\setR^n \colon \norm{x}_2\leq1\}$ under an \emph{affine transformation} $f(x) = E\,
x +d$ for some non-singular matrix $E \in \setR^{n×n}$ and a vector $d \in
\setR^n$, the problem whether such an ellipsoid contains a lattice vector
is the closest vector problem w.r.t. the $\ell_2$-norm in the lattice
$\Lambda(E^{-1} A)$ and target vector $E^{-1} d$.
Thus, we can apply the algorithm for $\op{CVP}_2$ to decide whether an
ellipsoid contains a lattice point or not. This gives us more
flexibility for the reduction of approximate $\op{CVP}_\infty$ to
$\op{CVP}_2$. 
Consequently, if we cover $B'$ with ellipsoids that are contained in
$B$ we can solve the approximate decision problem via calls to a
$\op{CVP}_2$-oracle.  This motivates the following covering problem.

\begin{quote}  
  How many ellipsoids that are contained in $[-1,1]^n$ are needed to
  cover $[-1+\varepsilon,1-\varepsilon]^n$?
\end{quote}

\noindent 
As we  mentioned above, the volume of the cube versus the volume of an
inscribed ball shows that covering with ellipsoids cannot yield a
singly-exponential dependence of the running time on the dimension
$n$. However, a similar idea and technique is the basis of our
promised complexity result.    The image of the unit-cube
$[-1,1]^n$ under  and affine transformation $f(x) = E \, x + d$ is a
\emph{parallelepiped}. 

With a $2$-approximation algorithm for
$\op{CVP}_\infty$ one can,  for a given parallelepiped $P$  find a lattice point in $P_s$, where $P_s$ stems from $P$ via scaling by $2$
  from its center of gravity $d$, or  assert that $P$ does not contain
  a lattice point. 
More precisely this can be done by a call to a $2$-approximation algorithm on the
lattice $\Lambda(E^{-1}A)$ and 
target-vector $E^{-1} d$. This motivates the following variant of the
above described covering problem.  

\begin{quote}
How many parallelepipeds that, if scaled by $2$ from their
centers of gravity are contained in the unit cube $[-1,1]^n$, are
necessary to cover the cube $[-1+\varepsilon , 1-\varepsilon]^n$? 
\end{quote}

\noindent 
We   consider the two covering problems from above and provide the
following results.
\begin{itemize}
\item  We show that the number of required ellipsoids is  bounded by
  $2^{O(n \log n)} (1+\log 1/\varepsilon)^n$ and provide a $c_n (1 + \lfloor\log
  1/\varepsilon\rfloor)^{n-1}$ lower bound for axis-parallel ellipsoids. 
\item We show that the number of required parallelepipeds is bounded
  from above by $2^n (1 + \log 1/\varepsilon)^n$ and from below by
  $c'_n (1+\lfloor\log 1/\varepsilon\rfloor)^n$. 
\end{itemize}
The second result yields a $2^{O(n)} (\log{1/\varepsilon})^{O(n)}$ randomized
algorithm that solves the approximate decision version of closest 
vector in the $\ell_\infty$-norm. The lower bound shows that this complexity
is optimal for an algorithm relying on this covering technique alone. 
Our main result, the $2^{O(n)} (\log
1/\varepsilon)^{O(n)}$ time $(1+\varepsilon)$-approximation algorithm, is then obtained via a
binary-search technique. We explain this in the final section of our
paper.

\section{The covering problems}
\label{sec:Covering}

We now consider the two covering problems from the introduction. We
denote the cube $[-1,1]^n$ by $H$ and its scaled version
$[-1+\varepsilon,1-\varepsilon]^n$ by $H_\varepsilon$. The questions are again as follows. 
Given an $\varepsilon\in(0,1)$, what is the smallest number $E(n,\varepsilon)$ 
of ellipsoids contained in $H$ such that their union covers the
smaller cube $H_\varepsilon$? What is the smallest number $P(n,\varepsilon)$ of
parallelepipeds that are contained in $H$ after being scaled by $2$
and whose union covers $H_\varepsilon$?

\subsection{Covering with ellipsoids}






We first show that $E(n,\varepsilon)$ is bounded by $2^{cn \log n}
(1+\log1/\varepsilon)^n$. Since we can allow us a factor of $2^n$, we cover
each intersection of $H_\varepsilon$ with an orthant separately and then
combine the different coverings, see also the figure on the
title-page. After flipping coordinates and after translation, the
problem for one orthant can be interpreted as follows.
How many ellipsoids that are contained in $H':=[0,2]^n$ are needed to
cover the cube $[\varepsilon,1]^n$? 

The following elementary lemma (see also Figure~\ref{fig:box-ball}) is
used in our construction. 
\begin{lemma} \label{lemma:box-ball} Let $n\geq 2$, $r = 1 +
  2/(\sqrt{n}-1)$ and $Q := [1/r,1]^n$, then the smallest ball
  containing  $Q$ is contained in $H'$. Furthermore,  $r$ is maximal
  with this   property.
\end{lemma}
\begin{proof}
  The center of $Q$ and  $B$ is  $d \cdot \bf{1}$ with
$$d = \frac{1+\frac1r}{2}=\frac{1+\frac{\sqrt{n}-1}{\sqrt{n}+1}}{2} = \frac{\sqrt{n}}{\sqrt{n}+1}.$$
Thus the radius $R$ of $B$ is simply the distance of $d \cdot \bf{1}$ to the vertices of $Q$
\[ R = \sqrt{n}(1-d) =
\sqrt{n}\cdot\left(1-\frac{\sqrt{n}}{\sqrt{n}+1}\right) =
\frac{\sqrt{n}}{\sqrt{n}+1} = d. \] Thus the ball is contained in the
positive orthant. Furthermore, $d+R < 2$, which shows the first claim,
i.e. that $B\subseteq H'$. 
The choice of $r$ is maximal because the ball touches the coordinate hyperplanes.
\end{proof}
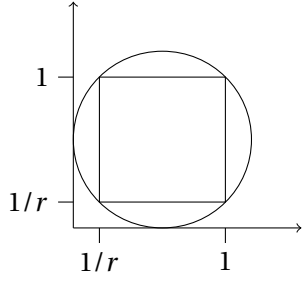
\begin{figure}
\begin{center}
  \begin{tikzpicture}[scale=2]
    \def\invr{0.1716}
    \def\R{0.5858}

    \draw[->] (0,0) -- (1.5,0);
    \draw[->] (0,0) -- (0,1.5);

    \foreach \x / \txt in {1/$1$,\invr/{$1/r$}}
      \draw (\x,0) -- (\x,-0.1) node[below] {\txt} (0,\x) -- (-0.1,\x) node[left] {\txt};

    \draw (\invr,\invr) rectangle (1,1);
    \draw (\R,\R) circle (\R);
  \end{tikzpicture}
\end{center}
\caption{Illustration of Lemma~\ref{lemma:box-ball}.}
\label{fig:box-ball}
\end{figure}

\begin{corollary} \label{corollary:box-ball}
  Let $n\geq 2$, $r = 1 +
  2/(\sqrt{n}-1)$,  $v\in(0,1]^n$ and let $Q := [v_1r^{-1},v_1]×\ldots×[v_nr^{-1},v_n]$.
  Then there exists an axis-parallel ellipsoid $E$ such that $Q\subseteq E\subseteq H'$.
\end{corollary}

This corollary is obtained from Lemma~\ref{lemma:box-ball} by
scaling.  We are now ready to prove the upper bound.

\begin{theorem}
  \label{th:UpperBound}One has 
  $E(n, \varepsilon) \leq 2^{cn\log n}\cdot \left(1+ \log1/\varepsilon\right)^n$ for a fixed constant $c>0$.
\end{theorem}
\begin{proof}
  We provide a covering of $[\varepsilon,1]^n$ by ellipsoids contained in $H' =
  [0,2]^n$. 
  Let $r=1+2/(\sqrt{n}-1)$ as in
  Corollary~\ref{corollary:box-ball}. The smallest ellipsoid
  containing a box of the form 
  $$
  Q(\alpha) = \left[r^{-(\alpha_1+1)},r^{-\alpha_1}\right] × \cdots ×
  \left[r^{-(\alpha_n+1)},r^{\alpha_n}\right], \quad \alpha\in\N_0^n
  $$
  is contained in $H'$. 
  How many of these boxes are needed to cover the cube $[\varepsilon,1]^n$? 

  It is enough to consider those boxes $Q(\alpha)$ with $r^{-\alpha_j}>\varepsilon$ for
  all $j$. Taking logarithms, one obtains $\alpha_j \log r < \log1/\varepsilon$. A
  standard approximation for the logarithm yields  $\log r > c'/\sqrt{n}$
  for some constant $c'>0$, and so we can conclude $\alpha_j <
  \sqrt{n}\log(1/\varepsilon)/c'$.  In total, we require at most
\[ 
  \left(1+\frac{\sqrt n}{c'}\cdot \log 1 / \varepsilon\right)^n \leq
  \left(\frac{\sqrt{n}}{c'}\right)^n (1+\log1/\varepsilon)^n \leq 2^{c n\log n} (1+\log1/\varepsilon)^n
\]
  boxes to cover $[\varepsilon,1]^n$. Since by
  Corollary~\ref{corollary:box-ball}, each of these boxes can be
  covered by an ellipsoid contained in $H'$, this completes the proof.
\end{proof}

\subsubsection{A lower bound for axis parallel ellipsoids}
\label{sec:LowerBound}

Can the dependence on $n$ be improved?  Note that the volume of $H_\varepsilon$
is $(2-2\varepsilon)^n$, whereas the largest ellipsoid contained in $H$ is the
$n$-dimensional euclidean ball with radius $1$ centered in $0$ which
is of volume $2^{-\Omega(n\log n)}$.  So for fixed $\varepsilon \in (0,\frac12)$,
simply by accounting for volume it is clear that we need at least
$2^{\Omega(n\log n)}$ ellipsoids.


What about the dependence on $\varepsilon$? This seems to be a more difficult
question. We can prove the following.
\begin{theorem} \label{thm:ellipsoid-lower-bound} Fix the dimension
  $n\geq 2$. There exists a constant $c_n>0$, depending only on $n$,
  such that for all $\varepsilon\in(0,1)$, any covering of $H_\varepsilon$ by axis
  parallel ellipsoids contained in $H$ consists of at least
  $c_n\cdot\left(1 + \left\lfloor\log 1 / \varepsilon \right\rfloor\right)^{n-1}$
  ellipsoids.
\end{theorem}
\begin{proof}
  To simplify the argument, we again transform the problem so that we can work entirely within the positive orthant. Consider the grid
  \[
    G_\varepsilon :=
      \{ v\in\R^n \mid v_j = 2^{-\alpha_j} \geq \varepsilon\text{ with } \alpha_j\in\N_0 \text{ for all } 1\leq j \leq n \}
  \]
  Every covering of $H_\varepsilon$ using $m$ axis parallel ellipsoids contained in $H$ corresponds, by an affine transformation, to a covering of $G_\varepsilon$ using axis parallel ellipsoids $E_1$,\dots,$E_m\subset\R^n_{\geq0}$. Because we can grow the ellipsoids until they touch all coordinate hyperplanes, we can assume without loss of generality that
  \[ E_i = \left\{ x\in\R^n \mid \sum_{j=1}^n \left(\frac{2^{-\mu_{ij}}-x_j}{2^{-\mu_{ij}}}\right)^2 \leq 1 \right\}, \]
  where the center of $E_i$ is at $(2^{-\mu_{i1}},\ldots,2^{-\mu_{in}})$.

  We will proceed to give an upper bound on the number $|E_i \cap G_\varepsilon|$ of grid points contained in an ellipsoid. Let $v = (2^{-\alpha_j})_{j=1}^n\in E_i \cap  G_\varepsilon$.
  \[
    1 \geq \sum_{j=1}^n \left( \frac{2^{-\mu_{ij}} - 2^{-\alpha_j}}{2^{-\mu_{ij}}} \right)^2
      = \sum_{j=1}^n \left( 1 - 2^{\mu_{ij} - \alpha_j} \right)^2
  \]
  At most one summand -- say the $k$-th -- can be greater than one half. We then must have $(1-2^{\mu_{ij} - \alpha_j})^2 \leq \frac{1}{2}$ for all $j\neq k$. A rough calculation shows $-2 < \mu_{ij} - \alpha_j < 1$, so there are at most $3$ possible choices of $\alpha_j\in\N_0$ for every $j\neq k$. On the other hand, $\alpha_k$ can take any integer value between $0$ and $\left\lfloor \log \frac1\varepsilon \right\rfloor$. Finally, there are $n$ choices for $k$, giving the upper bound of
  \[
    |E_i \cap G_\varepsilon| \leq n3^{n-1}\left(1+\left\lfloor\log 1 / \varepsilon \right\rfloor\right).
  \]
  Combining this with the total number of grid points, we get
  \[
    \left(1+\left\lfloor\log 1 / \varepsilon \right\rfloor\right)^n = |G_\varepsilon| \leq \sum_{i=1}^m |E_i\cap G_\varepsilon|
      \leq m n3^{n-1}\left(1+\left\lfloor\log 1 / \varepsilon \right\rfloor\right).
  \]
  The statement of the theorem follows, with $c_n = (n3^{n-1})^{-1}$.
\end{proof}
Note that this proof only works for axis parallel ellipsoids. It seems
implausible that allowing arbitrary ellipsoids could yield
significantly more efficient coverings.



\subsection{Covering with parallelepipeds}
\label{sec:CoveringWithParallelepipeds}

The goal is to cover $H_\varepsilon = [-1+\varepsilon, 1-\varepsilon]^n$ by parallelepipeds that,
if scaled by $2$, are contained in $H = [-1,1]^n$. The smallest number
of such parallelepipeds is $P(n,\varepsilon)$. We again provide an
axis-parallel covering. This time, however, we derive a lower
bound that is asymptotically  tight in the exponent,  even for
non-axis-parallel parallelepipeds.
We remark that the results of this sections hold with only minor numerical changes for
any constant scaling factor. We fix the factor $2$ for concreteness and to simplify the presentation.
First, we need an elementary
lemma whose proof is straightforward. See Figure~\ref{fig:CoveringWithParallelepipeds} for an illustration. 

\begin{lemma}
  \label{lem:1}
  Let $v \in (0,1]^n$ and $U =
  [1-v_1,1-v_1/3]×\dots×[1-v_n,1-v_n/3]$. If $U$ is scaled by a factor
  of $2$ from its center of gravity, then it is still contained in
  $[-1,1]^n$. 
\end{lemma}

\begin{theorem}
  \label{thr:1}
  One has $P(n,\varepsilon) \leq 2^n(1+\log 1/ \varepsilon)^n$.
\end{theorem}

\begin{proof}
  We proceed by covering  $[0,1-\varepsilon]^n$ by boxes
  that, if scaled by two, are contained in
  $[-1,1]^n$. 
  Consider a  box of the form 
  \[ U(\alpha) =
  \left[1-3^{-\alpha_1},1-3^{-\alpha_1-1}\right]×\dots×\left[1-3^{-\alpha_n},1-3^{-\alpha_n-1}\right], \quad \alpha\in\N_0^n. 
  \]
  By Lemma~\ref{lem:1} these boxes are still contained in $H$ after
  they are scaled by $2$. How many of these boxes are needed to cover
  $[0,1-\varepsilon]^n$?  We only have to consider $U(\alpha)$ with $3^{-\alpha_j}>\varepsilon$
  for all $j$.  Taking logarithms, this implies $\alpha_j <
  \frac{\log(1/\varepsilon)}{\log 3}$.  Thus we need at most $(1+\log1/\varepsilon)^n$
  boxes.  Repeating the procedure for each orthant yields the desired
  bound.
\end{proof}

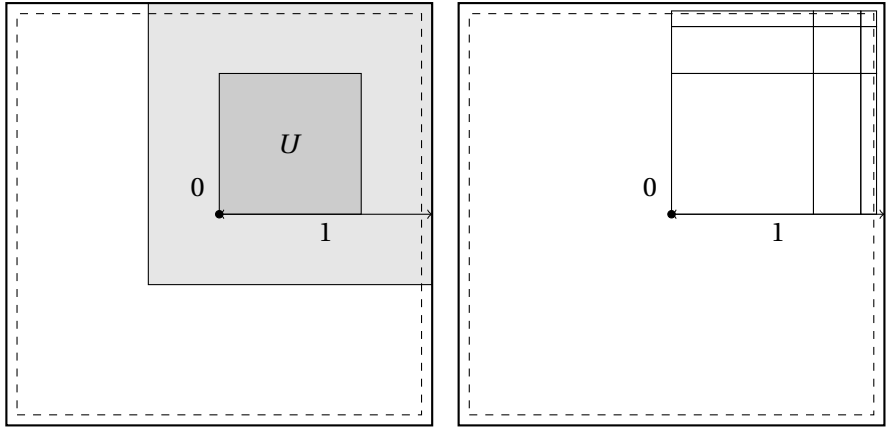
\begin{figure}
\begin{center}
  \begin{tikzpicture}[scale=0.7]
    \def\eps{0.2}
    \def\width{4}

    \draw[fill=black!10] (-\width/3,-\width/3) rectangle (\width,\width);
    \draw[fill=black!20] (0,0) rectangle (\width*2/3, \width*2/3);

    \draw (\width/3,\width/3) node {$U$};

    \node [label=120:$0$] {};
    \draw[fill=black] (0,0) circle (0.07);

    \draw[<->] (0.0,-0.0) -- node[auto, swap] {$1$} (\width,-0.0);

    \draw[thick] (-\width,-\width) rectangle (\width,\width);
    \draw[dashed] (-\width+\eps,-\width+\eps) rectangle (\width-\eps,\width-\eps);
  \end{tikzpicture}
~~
  \begin{tikzpicture}[scale=0.7]
    \def\eps{0.2}
    \def\width{4}

    \foreach \x in {\width,\width/3,\width/9}
      \foreach \y in {\width,\width/3,\width/9}
	\draw (\width-\x,\width-\y) rectangle (\width-\x/3,\width-\y/3);

    \node [label=120:$0$] {};
    \draw[fill=black] (0,0) circle (0.07);

    \draw[<->] (0.0,-0.0) -- node[auto, swap] {$1$} (\width,-0.0);

    \draw[thick] (-\width,-\width) rectangle (\width,\width);
    \draw[dashed] (-\width+\eps,-\width+\eps) rectangle (\width-\eps,\width-\eps);
  \end{tikzpicture}
\end{center}
\caption{Left: An illustration of Lemma~\ref{lem:1} for $v_1=\cdots=v_n=1$. Right: Covering one orthant with boxes of type $U(\alpha)$.}
\label{fig:CoveringWithParallelepipeds}
\end{figure}

\subsubsection{A lower bound}
The approach described in the previous section can be thought of, in a more general form, as the problem of covering the cube $H_\varepsilon$ using affine copies of a fixed centrally symmetric convex body $K$, such that constant multiples of the copies are still contained in $H$.
We will show that the number of parallelepipeds  is optimal as far as
the growth of the exponents are concerned.
The proof is analogous to that of Theorem~\ref{thm:ellipsoid-lower-bound}.
\begin{theorem}
  Let $K\subset\R^n$ be a centrally symmetric body. Let $K_1,\ldots,K_m$ be affine copies of $K$ and let $K'_j$ be the result of scaling $K_j$ by a factor of $2$ around its center point. Suppose that $K'_j\subseteq H$ for all $j$, and $K_1,\ldots,K_m$ together cover $H^\varepsilon$. Then $m \geq c_n \left(1+ \lfloor \log1/\varepsilon \rfloor \right)^n$, where $c_n > 0$ only depends on $n$.
\end{theorem}
\begin{proof}
  By translating the given bodies, we can instead consider a situation where $[\varepsilon,1]^n$ is covered by $K_1,\ldots,K_m$ and $K'_j\subset\R^n_{\geq0}$ for all $j$. In particular, this means that the grid
  \[
    G^\varepsilon :=
      \{ v\in\R^n \mid v_j = 2^{-\alpha_j} \geq \varepsilon\text{ with } \alpha_j\in\N_0 \text{ for all } 1\leq j \leq n \}
  \]
  is covered. Let us now determine the number of grid points contained in each $K_j$.
  Let $a_j$ be the center point of $K_j$. We have
\[ K_j \subseteq \{ x\in\R^n \mid \frac12a_j \leq x \leq \frac32 a_j \}, \]
  where the first set of inequalities follows from the fact that $K'_j\subset\R^n_{\geq0}$, and the second set of inequalities follows from central symmetry of $K_j$. There are at most two choices for $\alpha_i\in\N_0$ such that $x_i = 2^{-\alpha_i}$ satisfies the corresponding lower and upper bound. Consequently, $K_j$ contains at most $2^n$ grid points. Recall that the total number of grid points is $(1 + \lfloor \log1/\varepsilon \rfloor)^n$, from which the statement of the theorem follows.
\end{proof}

\section{The approximation algorithm}
\label{sec:ApproxCVP}

We now present our $(1+\varepsilon)$-approximation algorithm for the closest vector problem in the $\ell_\infty$ norm.
We describe a boosting technique that turns any constant factor approximation algorithm for $\CVP_\infty$ into a $(1+\varepsilon)$-approximation algorithm at the expense of
an additional factor of $2^{O(n)}\left(\log 1/\varepsilon \right)^{O(n)} b^{O(1)}$ in the running time, where $b$ denotes the encoding length of the input.
It is a Karp reduction approach, i.e. the constant factor approximation algorithm is used as an oracle and called multiple times on different inputs. 

We first consider the \emph{$\alpha$-gap } $\CVP_\infty$ problem, which is defined as follows. Given a lattice $\Lambda(A)$, a target vector
$t$ and a number $D>0$, either find a lattice vector $v\in \Lambda(A)$ with $\infnorm{v-t}\leq D$, 
or assert that all lattice vectors have distance more than $\alpha^{-1} D$.
We show how to construct a $(1+\varepsilon)$-gap algorithm for $\CVP_\infty$ from a $2$-gap algorithm
using the covering with parallelepipeds described in Section~\ref{sec:CoveringWithParallelepipeds}.

Afterwards we describe a binary search procedure to obtain a $(1+\varepsilon)$-approximation algorithm,
using the $(1+\varepsilon)$-gap algorithm as an oracle in each iteration of the binary search.

We plug the currently fastest known constant approximation solver, the Blömer and Naewe (BN) algorithm~\cite{bn2009}, into our construction and boost its success probability so that we obtain the following approximation algorithm.
\begin{theorem} \label{thm:cvp-approximation-algorithm}
  For every $\varepsilon\in(0,1)$, there is a randomized algorithm that $(1+\varepsilon)$-approximates $\op{CVP}_\infty$
  in time $2^{O(n)}(\log1/\varepsilon)^{O(n)}b^{O(1)}$ with success probability $1-2^{-\Omega(n)}$.
\end{theorem}
The randomness is due to the fact that the BN algorithm is randomized. Our construction is deterministic.

\subsection{Boosting gap solvers}
\label{sec:GapCVP}

We now describe the $(1+\varepsilon)$-gap algorithm for $\op{CVP}_{\infty}$ with the following properties.
\begin{theorem}
 \label{th:CVP_1plusepsfrom2approx}
  Given an oracle that solves $2$-gap $\op{CVP}_{\infty}$, 
  for every $\varepsilon\in (0,1]$ we can solve  $(1+\varepsilon)$-gap $\op{CVP}_{\infty}$ using at most
  $ 2^{n}\cdot \left(2 + \log 1 / \varepsilon\right)^n$
  oracle calls. 

The encoding size of instances for each oracle query are polynomial in $n$, the original encoding length and in $\log 1 / \varepsilon$.
\end{theorem}
In fact, any constant-gap oracle could be used.
We choose to fix the approximation factor to $2$ for concreteness and to simplify the presentation.

Let $(B,~t,~D)$ be the input. To solve the $(1+\varepsilon)$-gap problem, we either
need to find a vector $v\in \Lambda(B)$ with $\infnorm{v-t}\leq D$, or assert that the box
$$ T:=t+ D\cdot[-1+\delta, 1-\delta]^n$$ 
with $1-\delta = 1/(1+\varepsilon)$ does not contain a lattice point. 
By scaling the instance, we can assume without loss of generality that $D=1$.
Hence the box $T$ is a translate of the box $H_\delta=[-1+\delta, 1-\delta]^n$.
As discussed in Section~\ref{sec:CoveringWithParallelepipeds},
there is a covering of $H_\delta$ and therefore $T$ with singly exponential many parallelepipeds.
These parallelepipeds have the property that if they are scaled by a factor of $2$ around their center of gravity, then they are still contained within $t+[-1,1]^n$.
This is useful because with one call to a $2$-approximation oracle for $2$-gap $CVP_\infty$, we can either find a lattice vector with distance at most $1$ or
assert that one of the parallelepipeds does not contain a lattice vector, as we show in the following lemma.
\begin{lemma}
\label{lem:ParallelepipedTest}
 Given a lattice $\Lambda(A)$ and a parallelepiped $P:=\{ x\in \R^n~:~\infnorm{E(x-d)}\leq 1\}$, a single call to a $2$-gap oracle for $CVP_\infty$ 
either asserts that $P\cap\Lambda(A)=\emptyset$ or find a lattice vector $v\in \Lambda(A)$ contained in $P^s:=\{ x\in \R^n~:~\infnorm{E(x-d)}\leq 2\}$,
i.e. $P$ scaled by 2 around its center of gravity $d$.
\end{lemma}
\begin{proof}
 Define $B:=E\cdot A$ and $t:=E\cdot d$ and observe that $P\cap \Lambda(B)\neq \emptyset$ ($P^s\cap \Lambda(B)\neq \emptyset$) if and only if there is a vector $v\in \Lambda(B)$ with distance at most $1$ ($2$) from $t$.
\end{proof}

\begin{proof}[Proof of Theorem~\ref{th:CVP_1plusepsfrom2approx}]
Let $(B,~t,~D)$ be the input. By scaling the instance, we can assume without loss of generality that $D=1$.
  Let $\delta = \frac{\varepsilon}{1+\varepsilon}$, so that $1-\delta = \frac{1}{1+\varepsilon}$.
  Our goal is to either assert that the box $T = t + [-1+\delta,1-\delta]^n$
  is empty or to find a lattice vector in $t + [-1,1]^n$.

  Let $P_1, \ldots, P_k$ with $k\leq 2^n\cdot\left(2+\log(\frac{1}{\varepsilon})\right)^n$ be parallelepipeds as in Theorem~\ref{thr:1}.
Moreover for each $i$ let $P^s_i$ be the parallelepiped $P_i$ scaled by a factor of $2$ around its center of gravity. Then $P_1, \ldots, P_k$ cover $T$ and
  $P^s_i\subseteq t+[-1,1]^n$ for each $i$.
Lemma~\ref{lem:ParallelepipedTest} shows that for each $i$, a single call to the $2$-gap $CVP_\infty$  oracle either yields a lattice vector in $t + [-1,1]^n$ or
asserts that $P_i$ does not contain a lattice vector. Since the parallelepipeds cover $T$, if the answers for all oracle calls are negative, we can assert that $T$ 
does not contain a lattice vector.

Note that the encoding size of each parallelepiped is bounded by a polynomial in $n$ and $\log 1 / \varepsilon$, so the bound for the encoding size for each oracle call holds.
\end{proof}

\subsection{Approximating the closest vector problem}
\label{sec:ApproximatinAlgorithm}
In this section we first describe a procedure to Karp-reduce the problem of computing a $(1+\varepsilon)$-approximation for $\op{CVP}_\infty$ to $(1+O(\varepsilon))$-gap $\op{CVP}_\infty$. 
Then we combine our constructions with the BN algorithm to obtain the currently fastest (randomized) $(1+\varepsilon)$-approximation algorithm for $\op{CVP}_\infty$.

\begin{theorem} \label{thm:cvp-approximation-reduction}
For every $\varepsilon\in(0,1)$ and $\delta := \min\{\varepsilon/5, 1/2\}$, given access to a $(1+\delta)$-gap $\op{CVP}_{\infty}$ oracle,
one can compute a $(1+\varepsilon)$-approximation for $\op{CVP}_\infty$ using 
$O(\log b + \log n + \log1/\varepsilon)$ calls to the oracle.
\end{theorem}

We are given as input a lattice $\Lambda = \Lambda(A)$ and target vector $t$.
Let us assume that the distance $d(t,\Lambda)$ of a closest vector to the target vector is between $1$ and at most $2^{cn^2\cdot b}$ for some constant $c > 0$. 
This can be achieved by scaling, see~\cite{bn2009}.
We then perform a simple binary search in the following way:

\begin{enumerate}
  \item Set $\delta := \min\{\varepsilon/5, 1/2\}$
  \item Initialize $L \gets 0$ and $U \gets \left\lceil \log_{1+\delta} 2^{cn^2\cdot b} \right\rceil$.
  \item While $U-L \geq 3$, do a binary search step:
  \begin{enumerate}
    \item Solve the $(1+\delta)$-gap problem with input $(A,t,(1+\delta)^{L+\lceil(U-L)/2\rceil})$.
    \item If a lattice vector $v$ is returned, update $U \gets \left\lceil \log_{1+\delta} \norm{v-t}_\infty \right\rceil$.
    \item Otherwise, update $L \gets L + \lceil(U-L)/2\rceil - 1$.
  \end{enumerate}
  \item Solve the $(1+\delta)$-gap problem with input $(A,t,(1+\delta)^{U+1})$ and return the resulting lattice vector.
\end{enumerate}
We first prove the correctness of this procedure before we analyze its running time.

\begin{lemma}
\label{lemma:ApproxBS}
The algorithm from above has the following properties.
  \begin{enumerate}
    \item The binary search routine maintains the invariant that $(1+\delta)^L \leq d(t,\Lambda) \leq (1+\delta)^U$.
    \item The algorithm returns a lattice vector $v$ that satisfies $\norm{v-t}_\infty \leq (1+\varepsilon)d(t,\Lambda)$.
  \end{enumerate}
\end{lemma}
\begin{proof}
  \begin{enumerate}
    \item The initial choices of $L$ and $U$ are appropriate after scaling the lattice as mentioned in the beginning of this section.
      In the case 3(b), the existence of the lattice vector $v$ proves that the invariant is maintained by the update of $U$.
      In the case 3(c), that is, when the $(1+\delta)$-gap problem does not return a lattice vector,
      this implies by definition that $d(t,\Lambda) \geq (1+\delta)^{L+\lceil(U-L)/2\rceil-1}$ and so the invariant is maintained.

    \item In the end, we know that $d(t,\Lambda) \leq (1+\delta)^U$, so the final application of the $(1+\delta)$-gap problem is guaranteed to find a lattice vector $v$.
      This lattice vector satisfies
\[
  \norm{v-t} \leq (1+\delta)^{U+1} \leq (1+\delta)^{L+3} \leq (1+\delta)^3 d(t,\Lambda) \leq (1+5\delta) d(t,\Lambda) \leq (1+\varepsilon) d(t,\Lambda).
\]
      For the second inequality, we used the fact that $U-L$ is an integer and therefore $U-L\leq 2$. \qedhere
  \end{enumerate}
\end{proof}

\begin{proof}[Proof of Theorem~\ref{thm:cvp-approximation-reduction}]
Correctness of the procedure has already been shown in Lemma~\ref{lemma:ApproxBS}. It remains to bound the number of oracle calls.
  Let $M_j$ be the difference $U-L$ after the $j$-th search step.
  By the initial choices of $L$ and $U$ we have
\[ M_0 = \left\lceil \log_{1+\delta}2^{cn^2\cdot b} \right\rceil = \left\lceil \frac{cn^2b}{\log (1+\delta)} \right\rceil \leq c'n^2b/\delta \]
  for some constant $c'>0$.
  Let us analyze what happens in step $j$.
  In the case 3(b), we know that $\norm{v-t}_\infty \leq (1+\delta)^{\lceil L+ (U-L)/2 \rceil}$.
  Denoting the updated value of $U$ by $U'$, this implies $U' \leq L+\lceil (U-L)/2 \rceil$, and so $M_j \leq \lceil (U-L)/2 \rceil \leq M_{j-1}/2 + 1$.

  In the case 3(c), we get
\[ M_j = U - (L + \lceil (U-L)/2 \rceil - 1) = M_{j-1} - \lceil M_{j-1}/2 \rceil + 1 \leq M_{j-1}/2 + 1. \]
  We get the same upper bound in both cases and can conclude using induction that
\[ M_j \leq 2^{-j} M_0 + 1 + \frac{1}{2} + \frac{1}{4} + \dots \leq 2^{-j} M_0 + 2. \]
  This implies that the number of steps is bounded by $\lceil\log M_0\rceil$
  because the iteration stops when $M_j$ drops below $3$.
  From this we can derive the desired upper bound for the number of oracle calls.
\end{proof}

We now prove the main theorem by using the BN algorithm as a $2$-approximation and applying the boosting technique for gap $\CVP$ combined with the binary search procedure.
As their algorithm is randomized, one has to take care of the success probabilities.
Their algorithm has a failure probability of $2^{-\Omega(n)}$. Considering the amount
of oracle queries we have to issue and the requirement that \emph{every} call has to be successful, that failure probability is too high,
so we boost the success probability using standard techniques.

\begin{proof}[Proof of Theorem~\ref{thm:cvp-approximation-algorithm}]
  The BN algorithm has a success probability of at least $1-2^{-c\cdot n}$ 
  for some constant $c>0$ and a running time of $(2+\varepsilon')^{O(n)}\cdot b^{O(1)}$ when used as a $(1+\varepsilon')$-approximation algorithm.

  Set $a := c'\cdot\left\lceil  (1+\max\left\{\log \log 1 / \varepsilon,~1\right\}+\frac{1}{n}\log b)\right\rceil$ for an appropriate constant $c'>0$ that will be determined later.
  Let BN+ be an algorithm that runs BN as a $2$-approximation algorithm $a$ times on the same input
  and returns the closest vector that was found among all runs.
  This aggregated algorithm is a $2$-approximation algorithm with a running time of $\max\{\log \log\left(\frac{1}{\varepsilon}\right), 1\} \cdot 2^{O(n)} \cdot b^{O(1)}$ and success probability at least
$$ 1-2^{-acn} \geq 1-2^{-c'cn}\left(\log 1 / \varepsilon \right)^{-cc'n}\cdot b^{-cc'}.$$
  Using the boosting technique from Theorem~\ref{th:CVP_1plusepsfrom2approx}, we can construct a $(1+\delta)$-gap algorithm with $\delta := \min\{\varepsilon/5, 1/2\}$,
  using BN+ as a $2$-gap oracle. This amounts to a running time of
$ 2^{O(n)}\cdot \left(\log 1 / \varepsilon\right)^{O(n)}\cdot b^{O(1)}$
  for the $(1+\delta)$-gap algorithm.
  Plugging this as a black-box into the binary search procedure, we get a $(1+\varepsilon)$-approximation algorithm by Theorem~\ref{thm:cvp-approximation-reduction}.
  Moreover, the number of calls to the $(1+\delta)$-gap algorithm is bounded by
  $O(\log n + \log b + \log 1/\varepsilon).$
  Thus in total we get the desired running time bound of
$$2^{O(n)}(\log1/\varepsilon)^{O(n)}b^{O(1)}$$
  which is also an upper bound to the number of calls to BN+.
  The probability for failure of the $(1+\varepsilon)$-approximation algorithm is bounded by the probability that one of the runs of BN+ fails.
  By choosing $c'$ large enough, we get an upper bound of $2^{-\Omega(n)}$ for the failure probability from the union bound.

We remark that, although the encoding size for some of the instances we query the oracle for may exceed $b$, it always stays within $\op{poly}(n, \log\frac{1}{\varepsilon}, b)$, so the asymptotic running time
indicated above is not affected.
\end{proof}

\vspace{-10px}

\subsection*{Acknowledgment} 
We would like to thank  János Pach for discussions on coverings and for pointing us
to~\cite{er61}. 

\bibliography{references,papers,mybib}


\end{document}